\documentclass[a4paper,UKenglish]{lipics-v2016}
\usepackage{microtype}%if unwanted, comment out or use option "draft"

\usepackage[linesnumbered,ruled,vlined]{algorithm2e}
\usepackage{cleveref}

\newcommand{\ufp}{\mbox{$\mathsf{UFP}$}}
\newcommand{\ufpp}{\mbox{$\mathsf{UFPP}$}}
\newcommand{\ufpnba}{\mbox{$\mathsf{UFP}$-$\mathsf{NBA}$}}
\newcommand{\ufppnba}{\mbox{$\mathsf{UFPP}$-$\mathsf{NBA}$}}
\newcommand{\rufp}{\mbox{$\mathsf{Round}$-$\mathsf{UFP}$}}
\newcommand{\rufpp}{\mbox{$\mathsf{Round}$-$\mathsf{UFPP}$}}

\newcommand{\rufppnba}{\mbox{$\mathsf{Round}$-$\mathsf{UFPP}$-$\mathsf{NBA}$}}

\newcommand{\icp}{\mbox{$\mathsf{ICP}$}}
\newcommand{\rcol}{\mbox{$\mathsf{RCOL}$}}

\theoremstyle{plain}

\newtheorem{observation}[theorem]{Observation}

\newcommand{\argmin}{\mathrm{argmin}}

\newcommand{\rmax}{r_{\max}}
\newcommand{\cmax}{c_{\max}}

\newcommand{\cR}{\mathcal{R}}

\newcommand{\algtwoSp}{\mathtt{COL2SP}}
\newcommand{\algrcol}{\mathtt{RectCol}}

\newcommand{\Flarge}{F^L}
\newcommand{\Fsmall}{F^S}
\newcommand{\Fmid}{F^M}
\title{Improved Algorithms for Scheduling Unsplittable Flows on Paths\footnote{This work was partially supported by NSF grant CCF-1422715, a Google Research Award, and an ONR grant on network algorithms.}}
\titlerunning{Improved Algorithms for Scheduling Unsplittable Flows on Paths} %optional, in case that the title is too long; the running title should fit into the top page column

\author[1]{Hamidreza Jahanjou}
\author[2]{Erez Kantor}
\author[3]{Rajmohan Rajaraman}
\affil[1]{Northeastern University, Boston, MA, USA\\
  \texttt{hamid@ccs.neu.edu}}
\affil[2]{University of Massachusetts, Amherst, MA, USA\\
  \texttt{erez.kantor@gmail.com}}
\affil[3]{Northeastern University, Boston, MA, USA\\
  \texttt{rraj@ccs.neu.edu}}
\authorrunning{H. Jahanjou,  E. Kantor, and R. Rajaraman} %mandatory. First: Use abbreviated first/middle names. Second (only in severe cases): Use first author plus 'et. al.'

\Copyright{Hamidreza Jahanjou, Erez Kantor, and Rajmohan Rajaraman}%mandatory, please use full first names. LIPIcs license is "CC-BY";  http://creativecommons.org/licenses/by/3.0/

% mandatory: Please choose ACM 1998 classifications from http://www.acm.org/about/class/ccs98-html. E.g., cite as "F.1.1 Models of Computation". 
% mandatory: Please provide 1-5 keywords
\EventEditors{John Q. Open and Joan R. Access}
\EventNoEds{2}
\EventLongTitle{42nd Conference on Very Important Topics (CVIT 2016)}
\EventShortTitle{CVIT 2016}
\EventAcronym{CVIT}
\EventYear{2016}
\EventDate{December 24--27, 2016}
\EventLocation{Little Whinging, United Kingdom}
\EventLogo{}
\SeriesVolume{42}
\ArticleNo{23}
\begin{document}
\maketitle

\begin{abstract}
In this paper, we investigate offline and online algorithms for \rufpp, the problem of minimizing the number of rounds required to schedule a set of unsplittable flows of non-uniform sizes on a given path with non-uniform edge capacities.   \rufpp\ is NP-hard and constant-factor approximation algorithms are known under the no bottleneck assumption (NBA), which stipulates that maximum size of a flow is at most the minimum edge capacity.  We study \rufpp\ {\em without the NBA}, and present improved online and offline algorithms.  We first study offline \rufpp\ for a restricted class of instances called $\alpha$-small, where the size of each flow is at most $\alpha$ times the capacity of its bottleneck edge, and present an $O(\log(1/(1-\alpha)))$-approximation algorithm.  Our main result is an online $O(\log\log c_{\max})$-competitive algorithm for \rufpp\ for general instances, where $c_{\max}$ is the largest edge capacities, improving upon the previous best bound of $O(\log c_{\max})$ due to~\cite{epstein}.  Our result leads to an offline $O(\min(\log n, \log m, \log\log c_{\max}))$-approximation algorithm and an online $O(\min(\log m, \log\log c_{\max}))$-competitive algorithm for \rufpp, where $n$ is the number of flows and $m$ is the number of edges.
 \end{abstract}

\section{Introduction}
The {\em unsplittable flow problem on paths} (\ufpp) considers selecting a maximum-weight subset of flows to be routed simultaneously over a path while satisfying capacity constraints on the edges of the path.    
In this work, we investigate a variant of \ufpp\ known in the literature as \rufpp\ or {\em capacitated interval coloring}. The objective in \rufpp\ is to schedule {\em all}\/ the flows in the smallest number of rounds, subject to the constraint that the flows scheduled in a given round together respect edge capacities. Formally, in \rufpp\, we are given a path $P=(V,E)$, consisting of $m$ links, with capacities $\{c_j\}_{j\in [m]}$, and a set of $n$ flows $\mathcal{F}= \{f_i=(s_i, t_i, \sigma_i): i\in [n]\}$ each consisting of a source vertex, a sink vertex, and a size. A set $R$ of flows is feasible if all of its members can be scheduled simultaneously while satisfying capacity constraints. The objective is to partition $\mathcal{F}$ into the smallest number of feasible sets (rounds) $R_1, ..., R_t$. 
%Alternatively, partitioning can be seen as coloring where flows of the same color constitute a feasible set. For a set of flows $F$, we define its chromatic number, $\chi(F)$, to be smallest number of rounds (colors) into which $F$ can be partitioned.

%An obvious lower bound on $\chi(\mathcal{F})$ is $\lceil \rmax(\mathcal{F}) \rceil$ where $\rmax(\mathcal{F})$ is maximum edge congestion. The congestion of an edge $e_j$ defined as the ratio of the total size of flows using $e_j$ to its capacity $c_j$.

One practical motivation for \rufp\ is routing in optical networks. Specifically, a flow $f_i$ of size $\sigma_i$ can be regarded as a connection request asking for a bandwidth of size $\sigma_i$. Connections using the same communication link can be routed at the same time as long as the total bandwidth requested is at most the link capacity. Most modern networks have heterogeneous link capacities; for example, some links might be older than others. In this setting, each round (or color) corresponds to a transmission frequency, and minimizing the number of frequencies is a natural objective in optical networks.

A common simplifying assumption, known as the no-bottleneck assumption (NBA), stipulates that the maximum demand size is at most the (global) minimum link capacity; i.e. $\max_{i\in [n]} \sigma_i \leq \min_{j\in [m]} c_j$; most results on \ufpp\ and its variants are under the NBA (see \S\ref{prev}). A major breakthrough was the design of $O(1)$-approximation algorithms for the unsplittable flow problem on paths (\ufpp) without the NBA \cite{bonsama,mazing}. In this paper, we make progress towards an optimal algorithm for \rufpp\ {\em without} imposing NBA. 
%Note, however, that for a demand $d_i$ to be feasible, its size can be no more than its bottleneck capacity.

We consider both offline and online versions of \rufpp. In the offline case, all flows are known in advance. In the online case, however, the flows are not known {\em à priori} and they appear one at a time. Moreover, every flow must be scheduled (i.e. assigned to a partition) immediately on arrival; no further changes to the schedule are allowed.

Even the simpler problem \rufppnba, that is, \rufpp\ with the NBA, in the offline case, is $\mathbf{NP}$-hard since it contains Bin Packing as a special case (consider an instance with a single edge). On the other hand, if all capacities and flow sizes are equal, then the problem reduces to interval coloring which is solvable by a simple greedy algorithm.

\subsection{Previous work}\label{prev}
The unsplittable flow problem on paths (\ufpp) concerns selecting a maximum-weight subset of flows without violating edge capacities.  \ufpp\ is a special case of \ufp, the unsplittable flow problem on general graphs.  The term, \emph{unsplittable} refers to the requirement that each flow must be routed on a single path from source to sink. \footnote{Clearly, in the case of paths and trees, the term is redundant.  We use the terminology \ufpp\ to be consistent with the considerable prior work in this area.} \ufpp, especially under the NBA, \ufppnba, and its variants have been extensively studied \cite{temporalk,ARKIN19871,Salavatipour,Bansal-Epstein,Bar-Noy, Calinescu,DARMANN,Phillips,Chekuri2003}. Recently, $O(1)$-approximation algorithms were discovered for \ufpp\ (without NBA) \cite{bonsama,mazing}. Note that, on general graphs, \ufpnba\ is $\mathbf{APX}$-hard even on depth-3 trees where all demands are 1 and all edge capacities are either 1 or 2 \cite{Garg1997}.

\rufpp\ has been mostly studied in the online setting where it generalizes the interval coloring problem (\icp) which corresponds to the case where all demands and capacities are equal. In their seminal work, Kierstead and Trotter gave an optimal online algorithm for \icp\ with a competitive ratio of $3\omega - 2$, where $\omega$ denotes the maximum clique size \cite{KT}. Note that, since interval graphs are prefect, the optimal solution is simply $\omega$. Many works consider the performance of the first-fit algorithm on interval graphs. Adamy and Erlebach were the first to generalize \icp\ \cite{Adamy2004}. In their problem, interval coloring with bandwidth, all capacities are 1 and each flow $f_i$ has a size $\sigma_i \in (0,1]$. The best competitive ratio known for this problem is 10 \cite{AZAR200618,Epstein2005} and a lower bound of slightly greater than 3 is known \cite{linearityofFF}. The online \rufpp\ is considered in Epstein et. al. \cite{epstein}. They give a 78-competitive algorithm for \rufppnba, an $O(\log \frac{\sigma_{\max}}{c_{\min}})$-competitive algorithm for the general \rufpp, and lower bounds of $\Omega(\log\log n)$ and $\Omega(\log\log\log \frac{c_{\max}}{c_{\min}})$ on the competitive ratio achievable for \rufpp. In the offline setting, a 24-approximation algorithm for \rufppnba\ is presented in \cite{Elbassioni}.
%\footnote{The term was coined in the more general context of unsplittable flow problem (\ufp). Clearly, in the case of paths and trees, the term is redundant.}

\subsection{Our results}
We design improved algorithms for offline and online \rufpp.  Let $m$ denote the number of edges in the path, $n$ the number of flows, and $c_{\max}$ the maximum edge capacity. 

\begin{itemize}
\item 
In \S\ref{sec:small_flows}, we design an $O(log(1/(1-\alpha)))$-approximation algorithm for offline \rufpp\ for $\alpha$-small instances
%, $0 \le \alpha < 1$, 
in which the size of each flow is at most an $\alpha$ fraction of the capacity of the smallest edge used by the flow, where $0<\alpha<1$.  This implies an $O(1)$-approximation for any $\alpha$-small instance, with constant $\alpha$.  Previously, constant-factor approximations were only known for $\alpha \le 1/4$.  

\item 
In \S\ref{sec:large_flows}, we present our main result, an online 
$O(\log\log c_{\max}))$-competitive algorithm for general instances.  
This result leads to an offline $O(\min (\log n, \log m, \log\log c_{\max}))$-approximation algorithm and an online $O(\min(\log m, \log\log c_{\max}))$-competitive algorithm.  
\end{itemize}
Our algorithm for general instances, which improves on the $O(\log c_{\max})$-bound achieved in~\cite{epstein}, is based on a reduction to the classic rectangle coloring problem (e.g., see~\cite{MathScand,Kostochka,Chalermsook2011}).  We introduce a class of "line-sparse" instances of rectangle coloring that may be of independent interest, and show how competitive algorithms for such instances lead to competitive algorithms for \rufpp.  

Due to space limitations, we are unable to include all of the proofs in the main body of the paper; we refer the reader to the appendix for any missing proofs.

\section{Preliminaries}
In \rufpp\, we are given a path $P=(V,E)$ consisting of $m+1$ vertices and $m$ links, enumerated left-to-right as $v_0, e_1, v_1, ..., v_{m-1}, e_m, v_m$, with edge capacities $\{c_j\}_{j\in [m]}$, and a set of $n$ flows $\mathcal{F}= \{f_i=(s_i, t_i, \sigma_i): i\in [n]\}$, where $s_i$ and $t_i$ represent the two endpoints of flow $f_i$, and $\sigma_i$ denotes the size of the flow. Without loss of generality, we assume that $s_i < t_i$. We say that a flow $f_i$ uses a link $e_j$ if $s_i < j \leq t_i$. For a set of flows $F$, we denote by $F(e)$ and $F(j)$ the subset of flows in $F$ using edge $e$ and $e_j$ respectively.

\begin{definition}
The {\em bottleneck capacity} of a flow $f_i$, denoted by $b_i$,
is the smallest capacity among all links used by $f_i$ – such an edge is called the bottleneck edge for flow $f_i$.
\end{definition}

A set of flows $R$ is called {\em feasible}\/ if all of its members can be routed simultaneously without causing capacity violation. The objective is to partition $\mathcal{F}$ into the smallest number of feasible sets $R_1, ..., R_t$. A feasible set is also referred to as a {\em round}.
% That is, such a partition has to satisfy 
% \begin{equation}
% \forall k\in[t]\ \forall j\in [m] \sum_{f_i\in R_k \cap F(j)} \sigma_i \leq c_j.
% \end{equation}
Alternatively, partitioning can be seen as coloring where rounds correspond to colors. 
\begin{definition}
For a set of flows $F$, we define its {\em chromatic number}, $\chi(F)$, to be smallest number of rounds (colors) into which $F$ can be partitioned.
\end{definition}

\begin{definition}
The {\em congestion}\/ of an edge $e_j$ with respect to a set of flows $F$ is
\begin{equation}
r_j(F) = \frac{\sum_{i\in F(j) }\sigma_i}{c_j},
\end{equation}
that is, the ratio of the total size of flows in $F$ using $e_j$ to its capacity. Likewise $r_e(F)$ denotes the congestion of an edge $e$ with respect to $F$. Also, let $\rmax(F)= \max_j r_j(F)$ be the maximum edge congestion with respect to $F$. When the set of flows is clear from the context, we simply write $\rmax$.
\end{definition}
An obvious lower bound on $\chi(\mathcal{F})$ is maximum edge congestion; that is,
\begin{observation}
$\chi(\mathcal{F}) \geq \lceil \rmax(\mathcal{F}) \rceil$.
\end{observation}
\begin{proof}
Suppose $e_j$ is any edge of the path. In each round, the amount of flow passing through the edge is at most its capacity $c_j$. Therefore, the number of rounds required for the flows in $F$ using $e_j$ to be scheduled is at least $\lceil r_j(F) \rceil$.
\end{proof}

Without loss of generality, we assume that the minimum capacity, $c_{\min}$, is 1. Furthermore, let $c_{\max}=\max_{e\in E} c_e$ denote the maximum edge capacity. As is standard in the literature, we classify flows according to the ratio of size to bottleneck capacity.

\begin{definition}
\noindent Let $\alpha$ be a real number satisfying $0\leq \alpha \leq 1$. A flow $f_i$ is said to be  $\alpha$-{\em small} if $\sigma_i \leq \alpha\cdot b_i$ and $\alpha$-{\em large} if $\sigma_i > \alpha\cdot b_i$ (refer to Figure \ref{fig:Small_large} for an example). Accordingly, the set of flows $\mathcal{F}$ is divided into small and large classes
\begin{equation*} 
\Fsmall_{\alpha}=\{f \in F \mid f \mbox{ is } \alpha\mbox{-small} \};\;\;\;
\Flarge_{\alpha}=\{f \in F \mid f \mbox{ is } \alpha\mbox{-large} \}.
\end{equation*}
\end{definition}

\begin{figure}[htb]
\begin{center}
\includegraphics[scale=0.25]{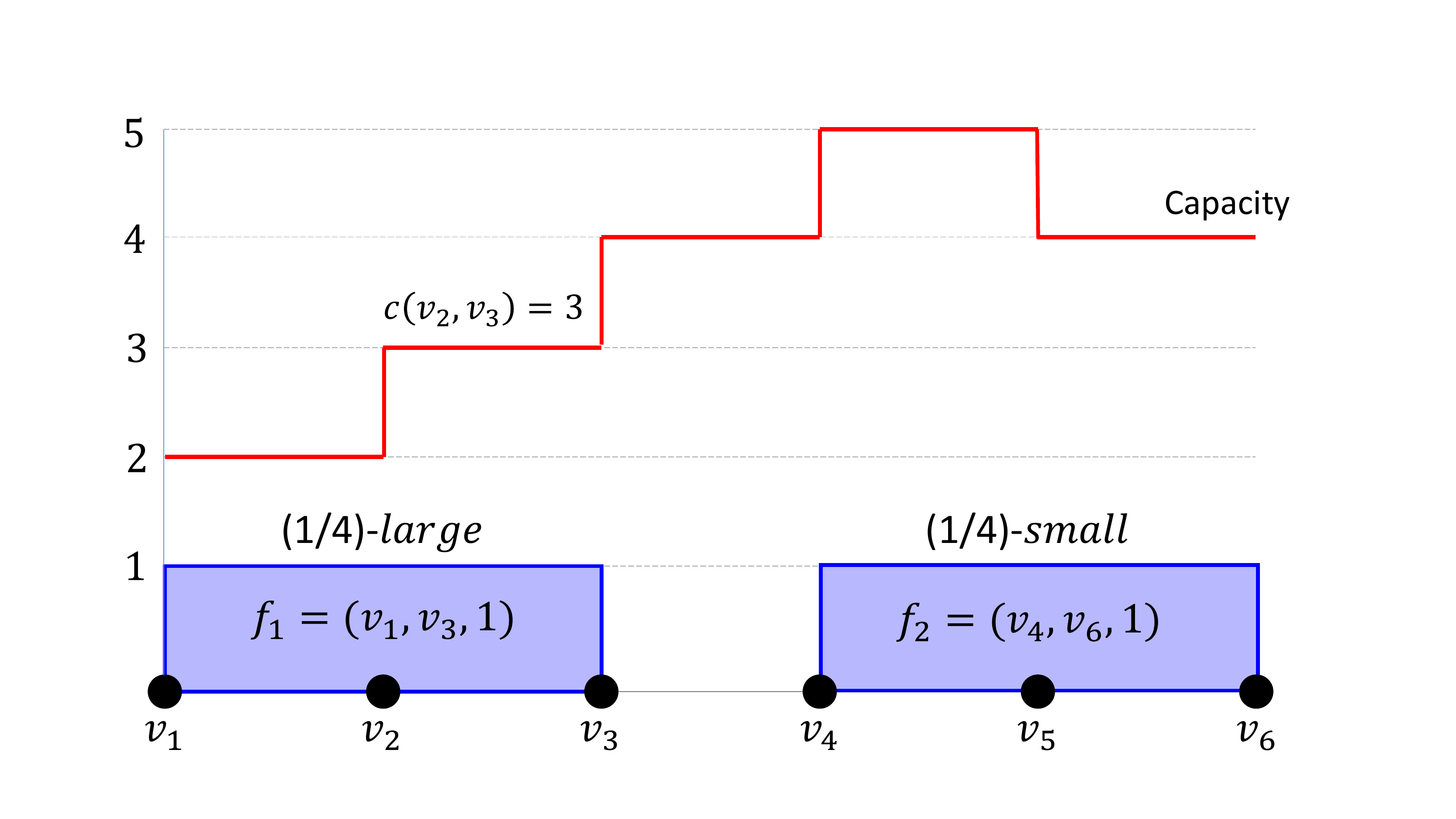}
\caption{An example of a path with 5 links and two flows. The first flow $f_1$ is from $v_1$ to $v_3$ of size 1; the second flow $f_2$ is from $v_4$ to $v_6$ also of size 1. Even though both flows have the same size, $f_1$ is $\frac{1}{4}$-large whereas $f_2$ is $\frac{1}{4}$-small. The reason is different bottleneck capacities, $b_1=2$ and $b_2=4$.
%That is, the bottlenecks-demand ratio of the first flow is 1/2 and hence $(1/4)$-$large$; and the bottlenecks-demand ratio of the second flow is 1/4 and hence $(1/4)$-$small$.
\label{fig:Small_large}}
\end{center}
\end{figure}

As is often the case for unsplittable flow algorithms, we treat small and large instances independently. In \S\ref{sec:small_flows} and \S\ref{sec:large_flows} we study small and large instances respectively.

\section{An approximation algorithm for \rufpp\ with $\alpha$-small flows}\label{sec:small_flows}
In this section, we design an offline $O(1)$-approximation algorithm for $\alpha$-small flows for {\em any} $\alpha \in (0,1)$.  We note that offline and online algorithms for $\alpha$-small instances are known when $\alpha$ is sufficiently small. More precisely, if $\alpha = 1/4$, $16$-approximation and $32$-competitive algorithms for offline and online cases have been presented in \cite{Elbassioni} and \cite{epstein} respectively. 

\begin{lemma}[\cite{Elbassioni,epstein}]
There exist $O(1)$-approximation algorithms for \rufpp\ where all flows are $\frac{1}{4}$-small.
\label{lemma: O(1)approx for small flows}
\end{lemma}

However, these results do not extend to the case where $\alpha$ is an arbitrary constant in $(0,1)$. In contrast, we present an algorithm that works for any choice of $\alpha\in (0,1)$.  In our algorithm, flows are partitioned according to the ratio of their size to their bottleneck capacity. If $\alpha\leq 1/4$, we simply use Lemma \ref{lemma: O(1)approx for small flows}. Suppose that $\alpha>1/4$. The overall idea is to further partition the set of flows into two subsets and solve each independently. This motivates the following definition.

\begin{definition}
Given two real numbers $0\leq \beta<\alpha<1$, a flow $f_i$ is said to be $[\beta,\alpha]$-{\em mid} if $\sigma_i \in[\beta\cdot b_i, \alpha\cdot b_i]$.  Accordingly, we define the corresponding set of flows as
\[
\Fmid(\beta,\alpha)=\{f \in F \mid f \mbox{ is } [\beta,\alpha]\mbox{-mid} \}. 
\]
Observe that, $\Fmid(\beta,\alpha)=\Fsmall_{\alpha} \cap \Flarge_{\beta}$.
\end{definition}

In the remainder of this section, we present an $O(1)$-approximation algorithm, called $\mathtt{ProcMids}$, for $\Fmid(1/4,\alpha)$.
$\mathtt{ProcMids}$ (see Algorithm \ref{algo:procmid}) starts by partitioning $\Fmid(1/4,\alpha)$ into $\lceil\log \cmax\rceil$ classes according to their bottleneck capacity.
% \begin{equation}\label{dfn: Fmid with levels}
% \Fmid_\ell=\{ f_i\in \Fmid(1/4,\alpha)\mid  2^{\ell-1}\leq b_i<2^{\ell}\},\; \text{for }\; \ell=1,...,\lceil\log \cmax\rceil.
% \end{equation}

\IncMargin{1em}
\begin{algorithm}
\SetKwInOut{Input}{input}\SetKwInOut{Output}{output}
\Input{A set of $[\frac{1}{4},\alpha]$-mid flows $F$}
\Output{A partition of $F$ into rounds (colors)}
\BlankLine
\For{$i\leftarrow 1$ \KwTo $\lceil \log \cmax \rceil$}{
$F_i\leftarrow\{ f_k\in F\mid  2^{i-1}\leq b_k<2^{i}\}$\;
$(C^i_1,C^i_2) \leftarrow \mathtt{FlowDec}(F_i)$\;
}
$R \leftarrow \mathtt{ColOptimize}(\{(C^k_1, C^k_2): k=1, ..., \lceil \log\cmax\rceil\})$\;
\Return $R$\;
\caption{$\mathtt{ProcMids}$}\label{algo:procmid}
\end{algorithm}\DecMargin{1em}

Next, it computes a coloring for each class by running a separate procedure called $\mathtt{FlowDec}$, explained in \S \ref{subsec: Procedure classschedule}.  This will result in a coloring of $\Fmid(1/4,\alpha)$ using $O(r_{\max} \log \cmax)$ colors. Finally, $\mathtt{ProcMids}$ runs $\mathtt{ColOptimize}$, described in \S\ref{subsec: unify colors}, to optimize color usage in different subsets; this results in the removal logarithmic factor and, thereby, a more efficient coloring using only $O(\rmax)$ colors. 
%In the analysis (see section \ref{}), we show that ... TBD...

\subsection{A logarithmic approximation}
\label{subsec: Procedure classschedule}
Procedure $\mathtt{FlowDec}$ (see Algorithm \ref{algo:flowdec} in Appendix \ref{Appendix:sec:missing codes}) partitions $\Fmid_\ell$ into $O(r_{\max}(\Fmid_\ell))$ rounds. In each iteration, it calls procedure $\mathtt{rCover}$ (Algorithm \ref{algo:rcover}) which takes as input a subset $F'_\ell\subseteq \Fmid_\ell$ and returns two disjoint feasible subsets $C_1,C_2$ of $F'_\ell$. In other words, flows in each subset can be scheduled simultaneously without causing any capacity violation. On the other hand, these two subsets cover all the links used by the flows in $F'_\ell$. More formally, $C_1$ and $C_2$ are guaranteed to have the following two properties:

\begin{itemize}
\item [(P1)] $\forall e \in E : |C_1(e)|\leq 1 \text{ and } |C_2(e)|\leq 1$,
\item [(P2)] $|F'_\ell(e)| > 1 \Rightarrow C_1(e)\cup C_2(e)\neq\emptyset$.
\end{itemize}

$\mathtt{rCover}$ maintains a set of flows $F''$ which is initially empty. It starts by finding the longest flow $f_{i_1}$ among those having the first (leftmost) source node. Next, it processes the flows in a loop. In each iteration, the procedure looks for a flow overlapping with the currently selected flow $f_{i_k}$. If one is found, it is added to the collection and becomes the current flow. Otherwise, the next flow is chosen among those remaining flows that start after the current flow's sink $t_{i_k}$. Finally, $\mathtt{rCover}$ splits $F''$ into two feasible subsets and returns them.

%Note that, in the $i$'th epoch of (give a name), Procedure (Erez: give a name) get $F^i_\ell$ and returns  $C^i_1$ and $C^i_2$.
%starts with $F^1_\ell=F_\ell$.  Let's denote by $C^i_1$ and $C^i_2$ the two subsets in the $i$'th epoch of $\classschedule$ and set $F^{i+1}_\ell=F^i_\ell\setminus(C^i_1\cup C^i_2)$.  Procedure $\classschedule$ stops at the first epoch $i'$ such that $F^{i'+1}_\ell=\emptyset$.  The pseudocode for $\classschedule$ is given below.
%

% $\classschedule(F_\ell)$:\\
% $\ell = 1$\\
% while $F^i_\ell \neq \emptyset$:\\
% $(C_1^i, C_2^i) \leftarrow \rcover(F^i_\ell)$\\
% $i \leftarrow i+1$\\
% $F_\ell^{i+1} \leftarrow F_\ell^i \setminus \left(C_i^i \cup C_2^i\right)$\\
% return $\{C_i^j, C_2^j: 1 \leq j < i\}$

\IncMargin{1em}
\begin{algorithm}
\SetKwInOut{Input}{input}\SetKwInOut{Output}{output}
\Input{A set of flows $F$}
\Output{Two disjoint feasible subsets of $F$ satisfying Properties (P1) and (P2)}
\BlankLine
$F''\leftarrow \emptyset$\;
$s_{\min}\leftarrow\min_{f_k\in F} s_k$\;
$t_{i_1}\leftarrow\max\{t_k \mid f_k\in F \text{ and } s_k=s_{\min}\}$\;
$F''\leftarrow \{f_{i_1}\}$\;
$F\leftarrow F \backslash \{f_{i_1}\}$\;
$k\leftarrow 1$\; 
\While{$t_{i_k} < \max_{f_i\in F}\{t_i\}$}{
\If{$\exists f_i\in F : s_i\leq t_{i_k} \text{ and } t_i > t_{i_k}$}{
$t_{i_{k+1}}\leftarrow \max\{t_i\mid f_i\in F \text{ and } s_i\leq t_{i_k}\}$\;
}
\Else{
$s_{\min}\leftarrow \min\{ s_i\mid f_i\in F,\ s_i > t_{i_k}\}$\;
$t_{i_{k+1}}\leftarrow \max\{t_i\mid f_i\in F,\ s_i = s_{\min}\}$\;
}
$F''\leftarrow F'' \cup \{f_{i_{k+1}}\}$\;
$k\leftarrow k+1$\;
}
$C_1\leftarrow \{f_{i_j} \in F'' \mid j \text{ is odd}\}$\;
$C_2\leftarrow \{f_{i_j} \in F'' \mid j \text{ is even}\}$\;
\Return $(C_1,C_2)$\;
\caption{$\mathtt{rCover}$}
\label{algo:rcover}
\end{algorithm}\DecMargin{1em}

\begin{lemma}\label{lem:rcover}
Procedure $\mathtt{rCover}$ finds two feasible subsets $C_1$ and $C_2$ satisfying properties (P1) and (P2).
%Given a flow set $F'_\ell \subseteq F_\ell$.
%Let $C_1$ and $C_2$ be the result of executing $\rcover$ on $F'_\ell$.
%Then,  holds.
\end{lemma}

\begin{lemma}
\label{classschedule}
Procedure $\mathtt{FlowDec}$ partitions $F^M_\ell$ into at most $8 \rmax(F^M_\ell)$ feasible subsets.
\end{lemma}

\subsection{Removing the $\log$ factor}
\label{subsec: unify colors}
In this subsection, we illustrate Procedure $\mathtt{ColOptimize}$ (see Algorithm \ref{algo:colopt}), which removes the logarithmic factor by optimizing color usage. The result is a coloring with $O(r_{\max})$ colors.
%(if a class has less than $2r_{\max}$ colors, then the remaining colors are empty). 
% The new colors are
% \begin{equation}
% D^i_a(k) = \bigcup_{z=0}^{\lceil(\log \cmax)/\tau\rceil-1} C^i_a(z\tau+k), \text{ for } a\in\{1,2\},\ k\in\{1,...,\tau\},\ i\in\{1,...,4r_{\max}\}.
% \end{equation}

\begin{algorithm}
\SetKwInOut{Input}{input}\SetKwInOut{Output}{output}
\Input{A set of pairs $\{(C^i_1(j),C^i_2(j))\}$, parameter $\tau$}
\Output{A new set of pairs $\{(D^i_1(j),D^i_2(j))\}$}
\BlankLine
\For{$i\leftarrow 1$ \KwTo $4\rmax$}{
\For{$k\leftarrow 1$ \KwTo $\tau$}{
$D^i_1(k) \leftarrow  \bigcup_{z=0}^{\lceil(\log \cmax)/\tau\rceil-1} C^i_1(z\tau+k)$\;
$D^i_2(k) \leftarrow  \bigcup_{z=0}^{\lceil(\log \cmax)/\tau\rceil-1} C^i_2(z\tau+k)$\;
}
}
\Return $\{(D^i_1(k), D^i_2(k)) : k=1, ..., \tau \text{ and } i\in\{1,...,4r_{\max}\}\}$\;
\caption{$\mathtt{ColOptimize}$}\label{algo:colopt}
\end{algorithm}\DecMargin{1em}

Let $\tau$ be a constant to be determined later. Intuitively, the idea is to combine subsets of different levels in an alternating manner with $\tau$ serving as the granularity parameter. More precisely, let $C^i_a(j)$, where $a\in\{1,2\}$, $j\in\{1,...,\lceil\log \cmax\rceil\}$, and $i\in\{1,...,4r_{\max}\}$, denote the set of colors resulting from the execution of $\mathtt{FlowDec}$. $\mathtt{ColOptimize}$ combines colors from different classes to reduce the number of colors by a factor of $\tau/\lceil\log \cmax\rceil$ resulting in $4\tau\cdot\rmax$ colors being used. An example is illustrated in Figure \ref{fig:combine-colors} in Appendix \ref{Appendix:sec:figs}. Next, we show that setting $\tau=\log(1/(1-\alpha)) +2$ results in a valid coloring.

\begin{lemma}\label{lem:validcol}
For $\tau=\log(1/(1-\alpha)) +2$, the sets $D^i_a(k)$, where $a\in\{1,2\}$, $k\in\{1,...,\tau\}$, and $i\in\{1,...,4r_{\max}\}$, constitute a valid coloring.
\end{lemma}

The main result of this section now directly follows from Lemma~\ref{lem:validcol}.

\begin{theorem}
For any $\alpha\in(0,1)$, there exists an offline $O(\log(1/1-\alpha))$-approximation algorithm for \rufpp\ with $\alpha$-small flows.  In particular, we have a constant-factor approximation for any constant $\alpha < 1$. 
\end{theorem}

\section{Algorithms for general \rufpp\ instances}\label{sec:large_flows}
In what follows, we present offline and online algorithms for general instances of \rufpp. Our treatment of large flows involves a reduction from \rufpp\ to the rectangle coloring problem (\rcol) which is discussed in \S\ref{subsec: The reduction}. Next, in \S\ref{subsec: Algorithms for rectangle coloring problem}, we design an online algorithm for the \rcol\ instances arising from the reduction. Later, in \S\ref{subsec: Online algorithm for rufpp on 1/4-large instance}, we cover our online algorithm for \rufpp\ with $\frac{1}{4}$-large flows. Finally, in \S\ref{sub:finalAlg}, we present our final algorithm for the general \rufpp\ instances.

% \begin{figure}[htb]
% \begin{center}
% \includegraphics[scale=0.25]{figs/Rec-R1.pdf}
% \caption{ A rectangle is specified by quadruple $(x^l(R), x^r(R), y^t(R), y^b(R))$.\label{fig:Rec-R.pdf.pdf}
% }
% \end{center}
% \end{figure}
% This problem is a special case of the graph coloring problem where vertices correspond to rectangles and edges correspond to intersections between the corresponding rectangles. We denote by $\omega(\cR)$ the maximum clique size of $\cR$ and by $\chi(R)$ its chromatic number.
% %
% Our algorithm has three steps.
% \begin{enumerate}
% \item Given the set of flows $F=\{f_1, ..., f_n\}$, construct a corresponding set of rectangles $\cR(F)=\{R_1, ..., R_n\}$.
% \item Solve the \rcol\ instance $\cR(F)$.
% \item Color each flow $f_i$ with the color of its corresponding rectangle $R_i$.
% \end{enumerate}
%The reduction in step 1, can be performed online. Consequently, an online (offline) algorithm used in step 2 gives rise to an online (offline) algorithm for \rufpp\ on $\frac{1}{4}$-large instances.

\subsection{The reduction from \rufpp\ with large flows to \rcol}
\label{subsec: The reduction}
\begin{definition}
{\bf Rectangle Coloring Problem (\rcol).}
Given a collection $\cR$ of $n$ axis-parallel rectangles, the objective is to color the rectangles with the minimum number of colors such that rectangles of the same color are disjoint.
\end{definition}

Each rectangle $R\in\cR$ is given by a quadruple $(x^l(R), x^r(R), y^t(R), y^b(R))$ of real numbers, corresponding to the $x$-coordinates of its left and right boundaries and the $y$-coordinates of its top and bottom boundaries, respectively. More precisely,
$R = \{ (x,y) \mid  x^l(R) \leq  x \leq x^r(R) \mbox{ and } y^b(R) \leq y \leq y^t(R)\}$. When the context is clear, we may omit $R$ and write $x^l, x^r, y^t, y^b$. Two rectangles $R$ and $R'$ are called compatible if they do not intersect each other; else, they are called incompatible.

The reduction from \rufpp\ with large flows to \rcol\ is based on the work in \cite{bonsama}. It starts by associating with each flow $f_i=(s_i,t_i,\sigma_i)$, a rectangle $R_i=(s_i,t_i,b_i,b_i-\sigma_i)$. If we draw the capacity profile over the path $P$, then $R_i$ is a rectangle of thickness $\sigma_i$ sitting under the curve touching the ``ceiling.'' Let $\cR(F)$ denote the set of rectangles thus associated with flows in $F$. We assume, without loss of generality, that rectangles do not intersect on their border; that is, all intersections are with respect to internal points. We begin with an observation stating that a disjoint set of rectangles constitutes a feasible set of flows.

\begin{observation}[\cite{bonsama}]\label{obser: coloring of RCOL is feasible for RUFP}
Let $\cR(F)$ be a set of disjoint rectangles corresponding to a set of flows $F$. Then, $F$ is a feasible set of flows.
\end{observation}

The main result here is that if all flows in $F$ are $k$-large then an optimal coloring of $\mathcal{R}(F)$ is at most a factor of $2k$ worse than the optimal solution to \rufpp\ instance arising from $F$. The following key lemma is crucial to the result.

\begin{lemma}[\cite{bonsama}]\label{lem:red}
Let $F$ be a feasible set of flows, and let $k\geq 2$ be an integer, such that every flow in $F$ is $\frac{1}{k}$-large. Then there exists a $2k$ coloring of $\cR(F)$.
\end{lemma}

As an immediate corollary, we get the following.

\begin{corollary}\label{cor:red}
Let $F$ be a feasible set of flows, and let $k\geq 2$ be an integer, such that every flow in $F$ is $\frac{1}{k}$-large. Then, $\chi(\cR(F)) \leq 2k\chi(F)$.
\end{corollary}
\begin{proof}
Consider an optimal coloring $C$ of $F$ with $\chi(F)$ colors. Apply Lemma \ref{lem:red} to each color class $C_i$, for $1\leq i\leq \chi(F)$, to get a $2k$-coloring of $\cR(C_i)$. The final result is a coloring of $\cR(F)$ using at most $2k\chi(F)$ colors.
\end{proof}

We are ready to state the main result of this subsection.

\begin{lemma}\label{lem:reduction}
Suppose there exists an offline $\alpha$-approximation (online $\alpha$-competitive) algorithm $\mathfrak{A}$ for \rcol. Then, for every integer $k\geq 2$ there exists an offline $2k\alpha$-approximation (online $2k\alpha$-competitive) algorithm for \rufpp\ consisting of $\frac{1}{k}$-large flows. 
\end{lemma}
\begin{proof}
Given a set $F$ of $\frac{1}{k}$-large flows for some integer $k\geq 2$, construct the set of associated rectangles $\cR(F)$ and apply the algorithm $\mathfrak{A}$ to it. The solution is a valid \rufpp\ solution (Observation \ref{obser: coloring of RCOL is feasible for RUFP}). Furthermore, by Corollary \ref{cor:red},
\[ \mathfrak{A}(\cR(F)) \leq \alpha \chi(\cR(F)) \leq 2k\alpha \chi(F). \]
Finally, the reduction does not depend on future flows; hence, it is online in nature.
\end{proof}

%The above proposition combined with the existence of $O(1)$-approximation algorithm for $\frac{1}{4}$-small instances means that the approximation ratio of an algorithm for \rufpp\ on $F$ is only a constant factor greater than the approximation ratio of \rcol\ problem on $\cR(F)$.
%
%
%The following corollary conclude (due to the above lemma and Observation \ref{obser: coloring of RCOL is feasible for RUFP }) that a \rcol\-based coloring approximation algorithm for \rufpp\ problem preserving the approximation ratio of the approximation algorithm for \rcol\ up to constant factor.

%\begin{corollary}[{\bf ``Preserving the approximation ratio up to constant factor''}]
%Given an $O(X)$-approximation algorithm $\RCOLAPP$ for \rcol\ problem.
%A \rcol\-based coloring algorithm which use $\RCOLAPP$ for \rufpp\ problem approximates the \rufpp\ problem with ratio of $O(X)$.
%\end{corollary}

\subsection{Algorithms for \rcol}
\label{subsec: Algorithms for rectangle coloring problem}

In this section, we consider algorithms for the rectangle coloring problem (\rcol). We begin by introducing a key notion measuring the sparsity of rectangles with respect to a set of lines. This is similar to the concept of point sparsity investigated by Chalermsook \cite{Chalermsook2011}.

\begin{definition}[s-line-sparsity]
A collection of rectangles $\cR$ is $s$-$line$-$sparse$ if there exists a set of axis-parallel lines $L_{\cR}$ (called an {\em $s$-line-representative set} of $\cR$), such that every rectangle $R\in\cR$ is intersected by $k_R\in [1,s]$ lines in $L_{\cR}$ (see Figure \ref{fig:sparsity} for an example).
\end{definition}

\begin{figure}[htb]
\begin{center}
\includegraphics[scale=0.25]{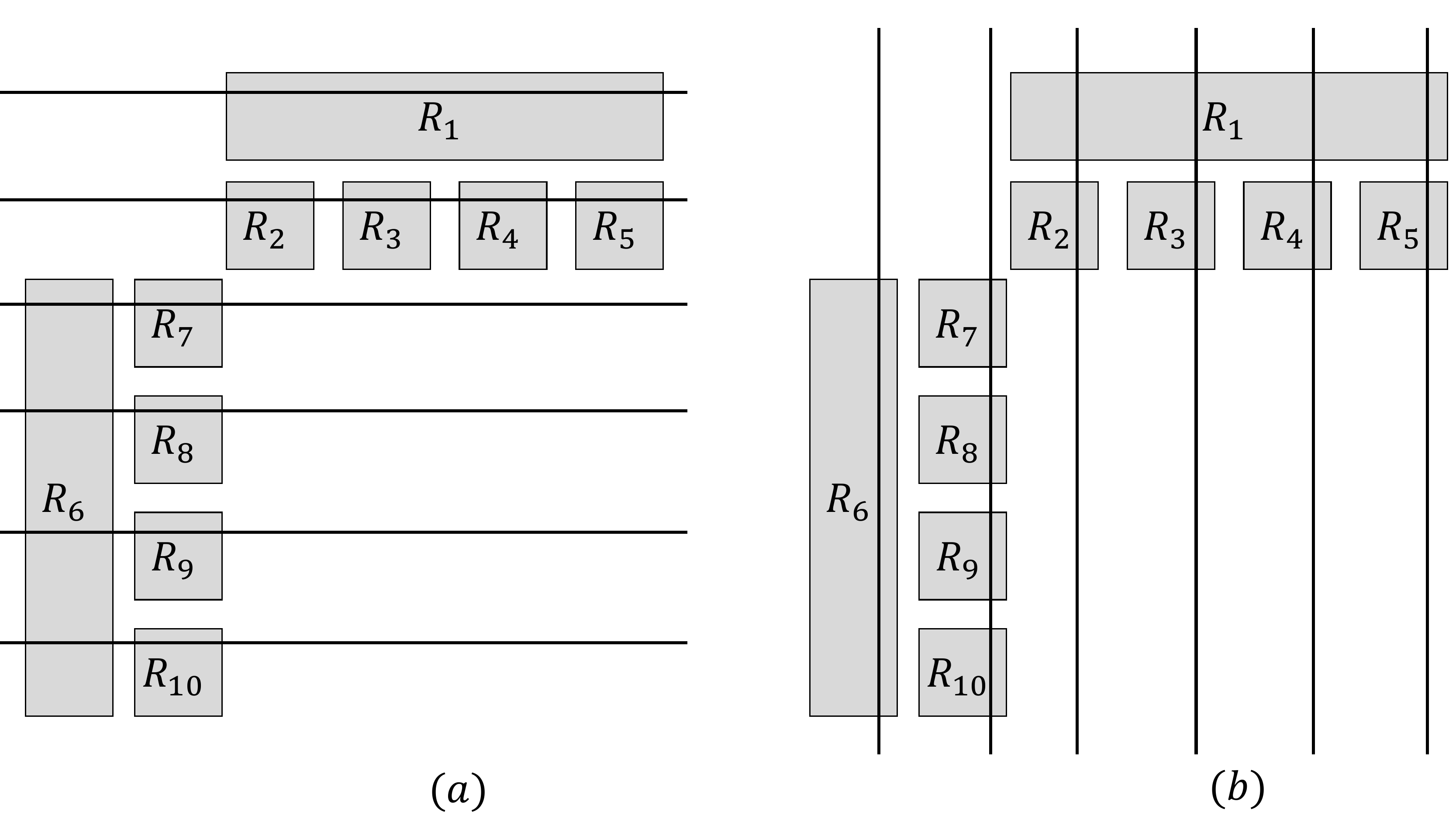}
\caption{A collection $\cR$ of 4-line-sparse rectangles. The lines can be either (a) horizontal or (b) vertical.
\label{fig:sparsity}
}
\end{center}
\end{figure}

For simplicity, we assume that representative lines are all horizontal. The objective is to design an online $O(\log s)$-competitive algorithm for \rcol\ consisting of $s$-line-sparse rectangles. In the online setting, rectangles appear one by one; however, we assume that an $s$-line-representative set $L_{\cR}$ is known in advance. As we will later see, this will not cause any issues since the \rcol\ instances considered here arise from \rufpp\ instances with large flows from which it is straightforward to compute $s$-line-representative sets. In the offline case, on the other hand, we get a $\log(n)$ approximation by (trivially) computing an $n$-line-representative set--associate to each rectangle an arbitrary line intersecting it. The remainder of this subsection is organized as follows. First, in \S\ref{subsub:2sp}, we consider the 2-line-sparse case. Later, in \S\ref{subsub:ssp}, we study the general $s$-line-sparse case.

\subsubsection{The 2-line-sparse case}\label{subsub:2sp}
Consider a collection of rectangles $\cR$ and a $2$-line-representative set $L_{\cR}=\{\ell_0,\ell_1,...,\ell_k\}$ (that is, each rectangle $R$ is intersected by either one or two lines in $L_{\cR}$) where the rectangles in $\cR$ appears in an online fashion. Recall, however, that the line set $L_{\cR}$ is known in advance. Without loss of generality, assume that $y(\ell_0)<y(\ell_1)<...<y(\ell_k)$.

For each $R\in \cR$, let $T(R)$ denote the index of the topmost line in $L_\cR$ that intersects $R$; $T(R)=\max\{ i \mid \ell_i \text{ intersects } R \}$. Next, partition $\cR$ into three subsets
\begin{equation} \label{eq:RL}
\cR_l = \{ R\in \cR \mid T(R)\equiv l\mod 3 \}, \text{ for } l=0,1,2.
\end{equation}
The following lemma shows that each of the above subsets can be viewed as a collection of interval coloring problem (\icp) instances.
\begin{lemma}\label{lem:sparse}
Suppose two rectangles $R,R'\in \cR$ belong to the same subset; that is, $R,R'\in \cR_l$ for some $l\in\{0,1,2\}$. Then, the following are true.
\begin{enumerate}
\item[(1)] If $T(R)=T(R')$ and the projection of $R$ and $R'$ on the $x$-axis have a non-empty intersection, then $R \cap R'\neq\emptyset$.

\item[(2)] If $T(R)\neq T(R')$, then $R \cap R'=\emptyset$.
\end{enumerate}
\label{lemma: rcol into interval coloring}
\end{lemma}

We will use the optimal 3-competitive online algorithm due to Kierstead and Trotter for \icp\ \cite{KT}. The algorithm colors an instance of \icp\ of clique size $\omega$ with at most $3\omega-2$ colors which matches the lower bound shown in the same paper. Henceforth, we refer to this algorithm as the $\mathtt{KT}$ algorithm.

Now we can present an $O(1)$-competitive online algorithm, named $\algtwoSp$, with a known $2$-line-representative set (see Algorithm \ref{algo:col2sp} in Appendix \ref{Appendix:sec:missing codes}). $\algtwoSp$ computes a partition of $\cR$ into $\cR_0, \cR_1$, and $\cR_2$ as explained above. Then, it applies the $\mathtt{KT}$ algorithm to each subset. Note that $\algtwoSp$ can be seen as executing multiple instances of the $\mathtt{KT}$ algorithm in parallel (see Figure \ref{fig:cR0} in Appendix \ref{Appendix:sec:figs}).

\begin{lemma}
Algorithm $\algtwoSp$ is an online $O(1)$-competitive algorithm for \rcol\ on 2-line-sparse instances given prior knowledge of a 2-line-representative set for the incoming rectangles. Moreover, $\algtwoSp$ uses at most $3\cdot\omega(\cR)$ colors% and this is optimal.
\label{lemma: Algorithm algtwoSp is O(1)-comp}
\end{lemma}

%(Recall that, the 2-representative-line-set property of $L_\cR$ implies that every rectangle $R\in\cR$ intersect %with at least one line and at most two lines of $L_\cR$.)

\subsubsection{The $s$-line-sparse case}\label{subsub:ssp}
Consider a set of $s$-line-sparse rectangles $\cR$ and an $s$-line-representative set $L_{\cR}$. Our goal in this subsection is to demonstrate a partitioning of $\cR$ into $O(\log s)$ $2$-line-sparse subsets, where each subset is accompanied by its own $2$-line-representative set.
%algorithm $\rcollinebase$ works in three steps.
%In the first step,  $\rcollinebase$ find (using Procedure XXX which we illustrate later) a {\em %s-representative-weak-line} set $L$ of $\cR$.
%Then, in the second step, it partition $\cR$ into $\logs=\lfloor \log s\rfloor+1$ sub-collection, where $s$ is the maximal number of lines of $L$ that intersect with a given rectangle $R$ of $\cR$.
%Finally, in the third step, it execute algorithm $\rcolgreedy$, separately,  on each of these sub-collections.
%
Given a set of lines $L$, we define the degree of a rectangle $R\in\cR$, with respect to $L$, to be the number of lines in $L$ that intersect $R$,
\[ \mathrm{Deg}_{L}(R)=\left| \{ \ell\in L \mid \ell\cap R\neq\emptyset \}\right|. \]
We say that a rectangle $R\in\cR$ is of level $l\geq 0$ with respect to $L_\cR$, if 
$2^l\leq \mathrm{Deg}_{L_{\cR}}(R)<2^{l+1}$. The partitioning is based on the level of rectangles. More precisely, $\cR$ is partitioned into $\lceil \log s\rceil+1$ ``levels"
\[\mathrm{Lev}(i)=\{ R\in\cR \mid R \mbox{ is of level } i \}, \mbox{ for } i=0,1,...,\lceil \log s\rceil.\]
Next we show that each level is a 2-line-sparse set. To this end, we present a 2-line-representative set for each level. Let $L_\cR=\{\ell_1,\ell_2,....,\ell_k\}$ and define 
\[S(i)=\{ \ell_j\in L_{\cR} \mid j \equiv 1 \mod 2^i \}, \text{ for } i\in\{0,...,\lceil \log s\rceil.\]
%denote the subset of $L_\cR$ which contains one out of every $2^l$ lines.

\begin{lemma}
For every $i\in\{0,...,\lceil \log s\rceil\}$, $\mathrm{Lev}(i)$ is a 2-line-sparse set and $S(i)$ is a 2-line-representative set for $\mathrm{Lev}(i)$.
\label{lemma: crlevel is 2-line-sparse}
\end{lemma}

We are ready to present an $O(\log s)$-competitive online algorithm, named $\mathtt{RectCol}$, for \rcol\  with a known line-representative set.
Algorithm $\algrcol$ works as follows (see Algorithm \ref{algo:rectcol}).

\IncMargin{1em}
\begin{algorithm}
\SetKwInOut{Input}{input}\SetKwInOut{Output}{output}
\Input{A rectangle $R\in\cR$}
\Input{The last state of $\mathtt{RectCol}$; an $s$-representative-line set $L_{\cR}$ for $\cR$}
\Output{A color for $R$}
\BlankLine
$i\leftarrow \argmin_j (2^j\leq\mathrm{Deg}_{L_{\cR}}(R)<2^{j+1})$\;
$\mathrm{Lev}(i)\leftarrow \mathrm{Lev}(i)\cup \{R\}$\;
\Return $\algtwoSp(R,L_{\cR})$\;
\caption{$\mathtt{RectCol}$}\label{algo:rectcol}
\end{algorithm}\DecMargin{1em}

% \begin{enumerate}
% \item Compute a partition of $\cR$ into $\mathrm{Lev}(0), ....,\mathrm{Lev}(\lceil \log s\rceil)$ in an online fashion 
% (which is trivially can compute, since the partition is depended on the rectangle coordinates and the representative-line coordinates);

% \item Compute $S(0), ....,S(\lceil \log s\rceil)$; and

% \item Execute the Algorithm $\algtwoSp$ for 2-line-sparse sets on $\mathrm{Lev}(i)$ with $S(i)$, for each $i=0,...,\lceil \log s\rceil$.
% \end{enumerate}

\begin{lemma} 
$\algrcol$ is an online $O(\log s)$-competitive algorithm for \rcol\ with s-line-sparse rectangles, given a representative-line set. Moreover, $\algrcol$ uses $O(\omega(\cR)\cdot\log s)$ colors.
\label{lemma: algrcol is O(log s)-comp}
\end{lemma}

\subsection{An algorithm for \rufpp\ with large flows}
\label{subsec: Online algorithm for rufpp on 1/4-large instance}
We are ready to present $\mathtt{ProcLarges}$, an algorithm for \rufpp\ with large flows.  For concreteness, we present the algorithm for $\frac{1}{4}$-large flows; this result can be easily generalized to $\alpha$-large flows for any $\alpha \le 1/2$. 
% %
% We are given the input for \rufpp\ on 1/4-large instances.
% We then illustarte a reduction transformation to the online version of \rcol\ with known representative-line-set. 
% Formally, for given a path $P=(V,E)$, consisting of $m$ links, with capacities $\{c_j\}_{j\in [m]}$, and a set of $n$ flows $\mathcal{F}= \{f_i=(s_i, t_i, \sigma_i): i\in [n]\}$ each consisting of a source vertex, a sink vertex, and a size. 
% %
% Suppose that $\mathcal{F}$ is (1/4)-large instance and its flows appears in online fasion. 
% %
% We assisate the above flow set with the rectangles collection denoted by $\cR(F)=\{R_1,R_2,...,R_n\}$.
% %
% For each flow $f_i=(s_i,d_i,\sigma_i)\in F$, denote by 
% $
% b_i=\min\{ c_j \mid s_i\leq j< d_i\}
% $
% the minimal capcity edge on the flow path. 
% %
% The rectangle $R_i$ specified by $(x_i^l, x_i^r, y_i^t, y_i^b)$, where $x_i^l=\source_i$, $x_i^r=\dest_i$, $y_i^t= b_i$, and $y_i^b=b_i-\sigma_i$ (as iluatrate in Subsection \ref{subsec: The reduction}).
% %
% Next, we define a representative-line-set for $\cR(F)$. 
The online algorithm we have designed for \rcol\ need to have access to an $s$-line-representative set $L_{\cR}$ for the set of rectangles $\cR$. In our case, these rectangles are constructed from flows (\S\ref{subsec: The reduction}) which themselves arrive in an online fashion. However, all we need to be able to compute an $s$-line-representative set is the knowledge of the path over which the flows will be running--that is $P=(V,E)$ with capacities $\{c_e\}_{e\in E}$ (recall that we assume that $c_{\min} = 1$, which can always be achieved via scaling if needed). It is possible to construct (at least) three  different $s$-line-representative sets for $\cR$: 

\begin{enumerate}
\item[$L_1$] A set of $s= \lceil  \log_{4/3} \cmax \rceil+1$ horizontal lines $L=\{l_0,l_1,...,l_s\}$ where the $y$-coordinate of the $i$th line is $y(l_i)= (3/4)^i\cdot\cmax$. Note that $\ell_0$ is the topmost line.
\item[$L_2$] A set of $m$ vertical lines, one per edge in the path.
\item[$L_3$] A set of $n$ axis-parallel lines, one per rectangle.
\end{enumerate}

Note that $L_3$ is only useful in the offline setting. It is obvious that $L_2$ and $L_3$ are valid line-representative sets for $\cR$. Below, we show that $L_1$ is valid as well.

\begin{lemma}\label{lem:l1}
$L_1$ is an $s$-line-representative set for $\cR(\mathcal{F})$.
\end{lemma}

% Our online algorithm for \rufpp\ with $\frac{1}{4}$-large flows, $\algruflarge$, works as follows.
% \begin{enumerate}

% \item computes the line-set $L$ (computes from the network parameters);

% \item when flow $f_i$ arrive do:
% 	\begin{enumerate}
% 	\item associate rectangle $R_i$ to the flow $f_i$ (assosiate as illustrated above);

% 	\item insert $R_i$ into the online fashion  execution of $\algrcol$ on $R_1,...,R_{i}$;

% 	\item color flow $f_i$ with the color of $R_i$.

% 	\end{enumerate}
% \end{enumerate}

The algorithm $\mathtt{ProcLarges}$, for $\frac{1}{4}$-large flows, can be seen in Algorithm \ref{algo:proclarges} in Appendix \ref{Appendix:sec:missing codes}.

\begin{theorem}\label{proc}
$\mathtt{ProcLarges}$ is an $O(\log\log \cmax)$-competitive algorithm for \rufpp\ with $\frac{1}{4}$-large flows. Furthermore, the bound can be improved to $O(\min(\log m, \log\log c_{\max}))$.
\end{theorem}
\begin{proof}
$\mathtt{ProcLarges}$ executes algorithm $\algrcol$ on $\cR(F)$ with a representative-line set $L=L_1$ of size $O(\log\cmax)$. The colors returned by $\algrcol$ are used for the flows without modification. Now, setting $s=O(\log\cmax)$, Lemma \ref{lemma: algrcol is O(log s)-comp} states that Algorithm $\algrcol$ uses $O(\omega(\cR(F))\log\log \cmax)$ colors. Lemma \ref{lem:reduction} completes the argument. Finally, note that running algorithm $\algrcol$ with $L=L_2$ as the representative-line set, we get a sparsity of $s=m$ and a coloring using $O(\omega(\cR(F))\log m)$ colors. To get the improved bound, we run the algorithm with $L=L_1$, if $\log\cmax\leq m$; else, we run it with $L=L_2$.
\end{proof}

\subsection{Putting it together -- The final algorithm} \label{sub:finalAlg}
At this point, we have all the ingredients needed to present our final algorithm ($\mathtt{SolveRUFPP}$--see Algorithm \ref{algo:rufp} in Appendix \ref{Appendix:sec:missing codes}) for \rufpp. $\mathtt{SolveRUFPP}$ simply uses procedure $\mathtt{ProcLarges}$ (\S\ref{subsec: Online algorithm for rufpp on 1/4-large instance}) for $\frac{1}{4}$-large flows and  procedure $\mathtt{ProcSmalls}$ for $\frac{1}{4}$-small flows. For $\mathtt{ProcSmalls}$, we can use our algorithm in \S\ref{sec:small_flows} or the 16-competitive algorithm in \cite{Elbassioni} in the offline case; and the $32$-competitive algorithm in~\cite{epstein} in the online case.

\begin{theorem}
There exists an online $O(\min(\log m, \log\log c_{\max}))$-competitive algorithm and an offline $O(\min (\log n, \log m, \log\log c_{\max}))$-approximation algorithm for $\rufpp$.
\end{theorem}
\begin{proof}
In the online case, $\mathtt{ProcSmalls}$ is a $32$-competitive \cite{epstein}. On the other hand, by Proposition \ref{proc}, $\mathtt{ProcLarges}$ is an $O(\min(\log m, \log\log c_{\max}))$-competitive. Thus overall, algorithm $\mathtt{SolveRUFPP}$ is $O(\min(\log m, \log\log c_{\max}))$-competitive. In the offline case, since the set of flows $\mathcal{F}$ is known in advance, we can get a slightly better bound by using $L_3$ in \S\ref{subsec: Online algorithm for rufpp on 1/4-large instance} as the third line-representative set (of sparsity $s=n$). Thus we get the $O(\min (\log n, \log m, \log\log c_{\max}))$ bound by running the algorithm three times with $L_1$, $L_2$, and $L_3$ and using the best one.
\end{proof}

\section{Concluding remarks}
In this paper, we present improved offline approximation and online competitive algorithms for \rufpp.  Our work leaves several open problems.  First, is there an $O(1)$-approximation algorithm for offline \rufpp?  Second, can we improve the competitive ratio achievable in the online setting to match the lower bound of $\Omega(\log\log\log c_{\max})$ shown in~\cite{epstein}, or improve the lower bound?  From a practical standpoint, it is important to analyze the performance of simple online algorithms such as  First-Fit and its variants for \rufpp\ and \rcol.  Another natural direction for future research is the study of \rufp\ and variants on more general graphs. 

\bibliography{refs}

\begin{thebibliography}{10}

\bibitem{AZAR200618}
An improved algorithm for online coloring of intervals with bandwidth.
\newblock {\em Theoretical Computer Science}, 363(1):18 -- 27, 2006.
\newblock Computing and Combinatorics.

\bibitem{Adamy2004}
Udo Adamy and Thomas Erlebach.
\newblock {\em Online Coloring of Intervals with Bandwidth}, pages 1--12.
\newblock Springer Berlin Heidelberg, Berlin, Heidelberg, 2004.

\bibitem{mazing}
Aris Anagnostopoulos, Fabrizio Grandoni, Stefano Leonardi, and Andreas Wiese.
\newblock A mazing 2+eps approximation for unsplittable flow on a path.
\newblock {\em CoRR}, abs/1211.2670, 2012.

\bibitem{ARKIN19871}
Esther~M. Arkin and Ellen~B. Silverberg.
\newblock Scheduling jobs with fixed start and end times.
\newblock {\em Discrete Applied Mathematics}, 18(1):1 -- 8, 1987.

\bibitem{MathScand}
E.~Asplund and B.~Grünbaum.
\newblock On a coloring problem.
\newblock {\em Mathematica Scandinavica}, 8(0):181--188, 1960.

\bibitem{Bansal-Epstein}
Nikhil Bansal, Amit Chakrabarti, Amir Epstein, and Baruch Schieber.
\newblock A quasi-ptas for unsplittable flow on line graphs.
\newblock STOC'06, pages 721--729, 2006.

\bibitem{Salavatipour}
Nikhil Bansal, Zachary Friggstad, Rohit Khandekar, and Mohammad~R.
  Salavatipour.
\newblock A logarithmic approximation for unsplittable flow on line graphs.
\newblock SODA'09, pages 702--709, 2009.

\bibitem{Bar-Noy}
Amotz Bar-Noy, Reuven Bar-Yehuda, Ari Freund, Joseph (Seffi)~Naor, and Baruch
  Schieber.
\newblock A unified approach to approximating resource allocation and
  scheduling.
\newblock {\em J. ACM}, 48(5):1069--1090, September 2001.

\bibitem{temporalk}
Mark Bartlett, Alan~M. Frisch, Youssef Hamadi, Ian Miguel, S.~Armagan Tarim,
  and Chris Unsworth.
\newblock The temporal knapsack problem and its solution.
\newblock CPAIOR'05, pages 34--48, Berlin, Heidelberg, 2005. Springer-Verlag.

\bibitem{bonsama}
P.~Bonsma, J.~Schulz, and A.~Wiese.
\newblock A constant factor approximation algorithm for unsplittable flow on
  paths.
\newblock In {\em FOCS'11}, pages 47--56, 2011.

\bibitem{Calinescu}
Gruia Calinescu, Amit Chakrabarti, Howard Karloff, and Yuval Rabani.
\newblock An improved approximation algorithm for resource allocation.
\newblock {\em ACM Trans. Algorithms}, 7(4):48:1--48:7, September 2011.

\bibitem{Chalermsook2011}
Parinya Chalermsook.
\newblock Coloring and maximum independent set of rectangles.
\newblock {\em APPROX'11}, pages 123--134, 2011.

\bibitem{Chekuri2003}
Chandra Chekuri, Marcelo Mydlarz, and F.~Bruce Shepherd.
\newblock Multicommodity demand flow in a tree.
\newblock In Jos C.~M. Baeten, Jan~Karel Lenstra, Joachim Parrow, and
  Gerhard~J. Woeginger, editors, {\em ICALP'03}, pages 410--425, 2003.

\bibitem{DARMANN}
Andreas Darmann, Ulrich Pferschy, and Joachim Schauer.
\newblock Resource allocation with time intervals.
\newblock {\em Theoretical Computer Science}, 411(49):4217 -- 4234, 2010.

\bibitem{Elbassioni}
Khaled~M. Elbassioni, Naveen Garg, Divya Gupta, Amit Kumar, Vishal Narula, and
  Arindam Pal.
\newblock Approximation algorithms for the unsplittable flow problem on paths
  and trees.
\newblock In {\em FSTTCS'12}, pages 267--275, 2012.

\bibitem{epstein}
Leah Epstein, Thomas Erlebach, and Asaf Levin.
\newblock Online capacitated interval coloring.
\newblock {\em SIAM Journal on Discrete Mathematics}, 23(2):822--841, 2009.

\bibitem{Epstein2005}
Leah Epstein and Meital Levy.
\newblock Online interval coloring and variants.
\newblock In {\em ICALP'05}, pages 602--613, 2005.

\bibitem{Garg1997}
N.~Garg, V.~V. Vazirani, and M.~Yannakakis.
\newblock Primal-dual approximation algorithms for integral flow and multicut
  in trees.
\newblock {\em Algorithmica}, 18(1):3--20, 1997.

\bibitem{linearityofFF}
H.~A. Kierstead.
\newblock The linearity of first-fit coloring of interval graphs.
\newblock {\em SIAM Journal on Discrete Mathematics}, 1(4):526--530, 1988.

\bibitem{KT}
H.~A. Kierstead and W.~T. Trotter.
\newblock An extremal problem in recursive combinatorics.
\newblock {\em Congressus Numerantium}, 33:143--153, 1981.

\bibitem{Kostochka}
Alexandr Kostochka.
\newblock Coloring intersection graphs of geometric figures with a given clique
  number.
\newblock In {\em Contemporary Mathematics 342, AMS}, 2004.

\bibitem{Phillips}
Cynthia~A. Phillips, R.~N. Uma, and Joel Wein.
\newblock Off-line admission control for general scheduling problems.
\newblock In {\em Journal of Scheduling}, pages 879--888, 2000.

\end{thebibliography}

\newpage
\appendix
\section{Missing proofs}
Here we provide the proofs that could not be included in the main text due to space constraints.

\begin{proof}[Proof of Lemma \ref{classschedule}]
By property (P1) of $\mathtt{rCover}$, $C^i_1$ and $C^i_2$ are both feasible sets for any $i$. Consequently, the flows in $C^i_1$ and $C^i_2$ can be scheduled in two rounds. Moreover, by property (P2) of $\mathtt{rCover}$, we have
\begin{equation}
r_{\max}(F^{i+1}_\ell) = r_{\max}(F^{i}_\ell\setminus(C^i_1\cup C^i_2))~\leq~r_{\max}(F^i_\ell)- 1/4~.
\end{equation}
This implies that $\mathtt{FlowDec}$ runs for at most $4\cdot r_{\max}(F_\ell)$ steps and therefore it partitions $F_\ell$ into at most $8\cdot r_{\max}(F_\ell)$ feasible subsets.
\end{proof}

\begin{proof}[Proof of Lemma \ref{lem:rcover}]
Let $F''_\ell = \{f_{i_1}, f_{i_2}, \ldots, f_{i_p}\}$ denote the set of flows obtained by $\mathtt{rCover}$ after termination of the loop.  We first establish that for $1 \le k < p$, $t_{i_k} < t_{i_{k+1}}$.  This follows immediately from the selection of $f_{i_{k+1}}$ in iteration $k+1$: in case 2, when there is an overlapping flow, the flow $f_i$ selected satisfies $t_i > t_{i_k}$, while in case 3, when there is no overlapping flow, the flow $f_i$ selected satisfies $t_i > s_i > t_{i_k}$.  

We next show that for $1 \le k < p$, if $k+2 \le p$, then $s_{i_{k+2}} \ge t_{i_k}$, which implies that no two flows in $C_1$ (resp., $C_2$) overlap, establishing property~(P1).  The proof is by contradiction.  Let $k$ be the smallest index that violates the preceding condition.  Consider iteration $k+1$.  Since $f_{k+2}$ satisfies the conditions $s_{i_{k+2}} \le t_{i_k}$ and $t_{i_k} \le t_{i_{k+2}}$, $f_{k+2}$ is a flow that satisfies the conditions of case 2 in iteration $k+1$.  Since $f_{i_{k+1}}$ is the flow selected in iteration $k+1$, it follows that $t_{k+1} \ge t_{k+2}$, a contradiction to the claim we have just established.

It remains to establish property (P2).  The proof is again by contradiction.  Let $e$ be the left-most edge for which~(P2) is violated, and let $f_j$ be a flow that uses edge $e$.  We consider two cases.  The first case is where there exists an index $k$ such that $s_j \le t_{i_k}$.  In this case, in iteration $k+1$, $f_j$ is a flow such that $s_j \le t_{i_k}$ and $t_j > t_{i_k}$ (the latter holds, since otherwise $e$ is covered by flow $f_{i_k}$ leading to a contradiction).  So the overlapping flow condition of $\mathtt{rCover}$ holds; therefore, $s_{i_{k+1}} \le t_{i_k}$ and $t_{i_{k+1}} \ge t_j$, implying that $e$ is covered by flow $f_{i_{k+1}}$, leading to a contradiction.  The second case is where there is no index $k$ such that $s_j \le t_{i_k}$; in particular $s_j > t_{i_\ell}$ where $\ell$ is the number of flows in $C_1 \cup C_2$.  This leads to another contradiction since the termination condition implies $t_{i_\ell} \ge t_j$.  This  establishes property~(P2) and completes the proof of the lemma.
\end{proof}

\begin{proof}[Proof of Lemma \ref{lem:validcol}]
Fix an edge $e\in E$. Let $F(e)$ be the set of flows in $F$ that use $e$. We need to show that $\{D^i_a(k)\}_{i,a,k}$ respect the capacity of $e$. Let $L(e)= \lceil \log c(e) \rceil$.
Note that the flows in $F(e)$  belong to $\Fmid_t$ for $t\leq L(e)$. In other words, $F(e)\cap\Fmid_t = \emptyset$, for every $t>L(e)$.

Recall that, by Property (P1), $\mathtt{FlowDec}$ assigns in each level at most one flow that uses $e$ to each color. This means that $D^i_a(k)$ has at most one flow that uses $e$ from each set $C^i_a(z\tau+k)$, for $z=0,...,\lceil L(e)/\tau\rceil-1$.
Moreover, the size $\sigma_i$ of a flow $f_i\in \Fmid_t$ is bounded by $\alpha\cdot 2^t$.
On the other hand, each flow $f_i\in\Fmid_{L(e)}$ satisfies $\sigma_i\leq \alpha\cdot c_e$ since it is $\alpha$-small
Thus, the total size of the flows in $D^i_a(k)$ that go through $e$ is
\begin{eqnarray*}
\alpha \cdot c_e + 2^{L(e)-\tau} +  2^{L(e)-2\tau}+... &\leq& \alpha\cdot c_e + 2^{L(e)-\tau}(1+1/2+1/2^2+....)\\
&=& \alpha\cdot c_e + 2^{L(e)-\tau+1}\\ 
&\leq& \alpha\cdot c_e + 2^{-\tau+2} c_e 
\leq c(e),
\end{eqnarray*}
where the third inequality follows from $c(e) >2^{L(e)-1}$; and the last inequality from $\tau = \log (1/(1-\alpha)) +2$, equivalently $2^{2-\tau}\leq (1-\alpha)$.
\end{proof}

\begin{proof}[Proof of Lemma \ref{lem:sparse}]
(1) is easy to verify. Indeed, the projections of $R$ and $R'$ on the $y$-axis both contain $y(\ell_{T(R)})$; hence, their intersection is non-empty. Thus, $R$ and $R'$ intersect if and only if their projection on the $x$-axis has a non-empty intersection.

Next, we prove (2). Consider two rectangle $R,R'\in \cR_l$, where $T(R)\neq T(R')$.
Let $i=T(R)$ and $i'=T(R')$. Assume, without loss of generality, that $i<i'$. Note that $i'\geq i+3$ by definition. Additionally, $y^t(R)<y(\ell_{i+1})$ since $\ell_i$ is the topmost line of $L$ that intersects $R$. On the other hand, $y^b(R')>y_{\ell_{i+1}}$ since $L$ is a 2-line-representative set of $\cR$ meaning that at most two lines in $L$  intersect $R'$. Consequently, the projection of $R$ and $R'$ on the $y$-axis have an empty intersection. Therefore, the $R$ and $R'$ do not intersect.
\end{proof}

\begin{proof}[Proof of Lemma \ref{lemma: Algorithm algtwoSp is O(1)-comp}]
Let $\mathrm{}{Rec}(\ell_i)$ denote the set of rectangles for which the line $\ell_i$ is the topmost line intersecting it. More precisely,
\[
\mathrm{Rec}(\ell_i) = \{ R\in \cR \mid T(R)=i \},\; \text{ for } i=0,1,...,k.
\]
Observe that, $\cR_l$ defined in (\ref{eq:RL}), satisfies
\[
\cR_l=\bigcup_{j=0}^{\lfloor(k-l)/3\rfloor}\mathrm{Rec}(\ell_{3j+l}),\; \text{ for } l = 0, 1, 2.
\]
%Denote by $\omega(\cR')$ the clique size of $\cR'$ for every $\cR'\subseteq\cR$. 
% 
Now, executing the $\mathtt{KT}$ algorithm on $\cR_l$, is equivalent to executing the $\mathtt{KT}$ algorithm on $\mathrm{Rec}(\ell_{l})$, $\mathrm{Rec}(\ell_{3+l})$, $\mathrm{Rec}(\ell_{6+l})$, ..., simultaneously. Indeed, by Lemma \ref{lemma: rcol into interval coloring}, for every $R,R'\in\cR_l$, we know that $R\cap R'=\emptyset$ if $R\in\mathrm{Rec}(\ell_i)$, $R'\in\mathrm{Rec}(\ell_j)$ and $i\neq j$. On the other hand, if $R,R'\in\mathrm{Rec}(\ell_i)$, part (2) of the lemma implies that the problem of coloring $\mathrm{Rec}(\ell_i)$ is the same as to that of coloring intervals resulting from the projection of $\mathrm{Rec}(\ell_i)$ on the $x$-axis.
Finally, since the $\mathtt{KT}$ algorithm is 3-competitive, $\algtwoSp$ uses at most $3\omega(\mathrm{Rec}(\ell_i))$ colors to color $\mathrm{Rec}(\ell_i)$. Hence, overall, $\algtwoSp$ colors $\cR_l$ with at most
\[
3\cdot \max\{\omega(\cR(\ell_i)) \mid i=l,3+l, 2\cdot 3+l,...,\lfloor(k-l)/3\rfloor\cdot 3+l \}\leq 3\cdot \omega(\cR_l) \leq 3\cdot \omega(\cR)
\] colors for $l=0,1,2$.
% It remains to prove the optimality which follows easily from the well-known lower bound of $3\omega - 2$ for \icp\ shown in \cite{KT}. 
\end{proof}

\begin{proof}[Proof of Lemma \ref{lemma: crlevel is 2-line-sparse}]
Fix an $i\in\{0,...,\lceil \log s\rceil\}$. 
If $\mathrm{Lev}(i)=\emptyset$, then it trivially is 2-line-sparse and any set of lines can serve as its 2-line-representative set. Now, suppose that $\mathrm{Lev}(i)\neq\emptyset$ and pick an arbitrary rectangle $R\in \mathrm{Lev}(i)$.
We need to show that $R$ intersects exactly either one or two lines in $S(i)$. By definition, we have that $2^i\leq \mathrm{Deg}_{L_{\cR}}(R)<2^{i+1}$. On the other hand, $S(i)\subseteq L_{\cR}$ contains one line for every $2^i$ lines of $L_{\cR}$. Hence, $R$ intersects at least one line and at most two lines in $\mathrm{Lev}(i)$.
\end{proof}

\begin{proof}[Proof of Lemma \ref{lemma: algrcol is O(log s)-comp}]
Consider an $s$-line-sparse set of rectangles $\cR$ and an $s$-line-representative set $L$.
By Lemma \ref{lemma: crlevel is 2-line-sparse}, $\mathrm{Lev}(i)$ is 2-line-sparse and $S(i)$ is a 2-line-representative set of $\mathrm{Lev}(i)$, for each $i=0,...,\lceil \log s\rceil$. Let $\#(\algtwoSp, \mathrm{Lev}(i))$ denote the number of colors used by algorithm $\algtwoSp$ to color $\mathrm{Lev}(i)$. Observe that $\algrcol$ use at most $\sum_{i=0}^{\lceil \log s\rceil} \#(\algtwoSp, \mathrm{Lev}(i))$
colors. Furthermore, by Lemma \ref{lemma: Algorithm algtwoSp is O(1)-comp}, $\#(\algtwoSp, \mathrm{Lev}(i)) \leq 3\omega(\cR)$, for every $i=0,...,\lceil \log s\rceil$. Therefore, Algorithm $\algrcol$ uses at most $3(\lceil\log s\rceil+1)\omega(\cR)$ colors.
\end{proof}

\begin{proof}[Proof of Lemma \ref{lem:l1}]
Since there are exactly $s$ lines in $L_1$, every rectangle in $\cR(F)$ is intersected by at most  $s$ lines. It remains to show that every rectangle is intersected by at least one line in $L_1$. To this end, consider an arbitrary rectangle $R_i\in\cR(F)$. Since $R_i$ corresponds to a $\frac{1}{4}$-large flow $f_i$, we have that $\sigma_i\geq b_i/4= y_i^t/4$, where $y^t_i$ is the top $y$-coordinate of the rectangle. Now, let $j$ be an index such that 
\begin{equation} \label{eq:ys}
y(l_{j+1}) < y_i^t\leq y(l_j). 
\end{equation}
Note that such an index exists, since $y(l_0)=\cmax$ and $y(l_s) < c_{\min}=1$.
It follows from the right-hand side of (\ref{eq:ys}) that $\frac{3}{4} y_i^t \leq \frac{3}{4} y(l_j)$. On the other hand, $y(l_{j+1})=\frac{3}{4} y(l_j)$ by definition. Furthermore, $y_i^b < \frac{3}{4} y_i^t $ since $\sigma_i \geq y_i^t/4$. Therefore, 
\[
y_i^b\leq y(l_{j+1}) < y_i^t, 
\]
which implies that $l_{j+1}$ intersects $R_i$. This completes the proof.
\end{proof}

\clearpage
\section{Missing figures}
\label{Appendix:sec:figs}
\begin{figure}[h]
\begin{center}
\includegraphics[scale=0.37]{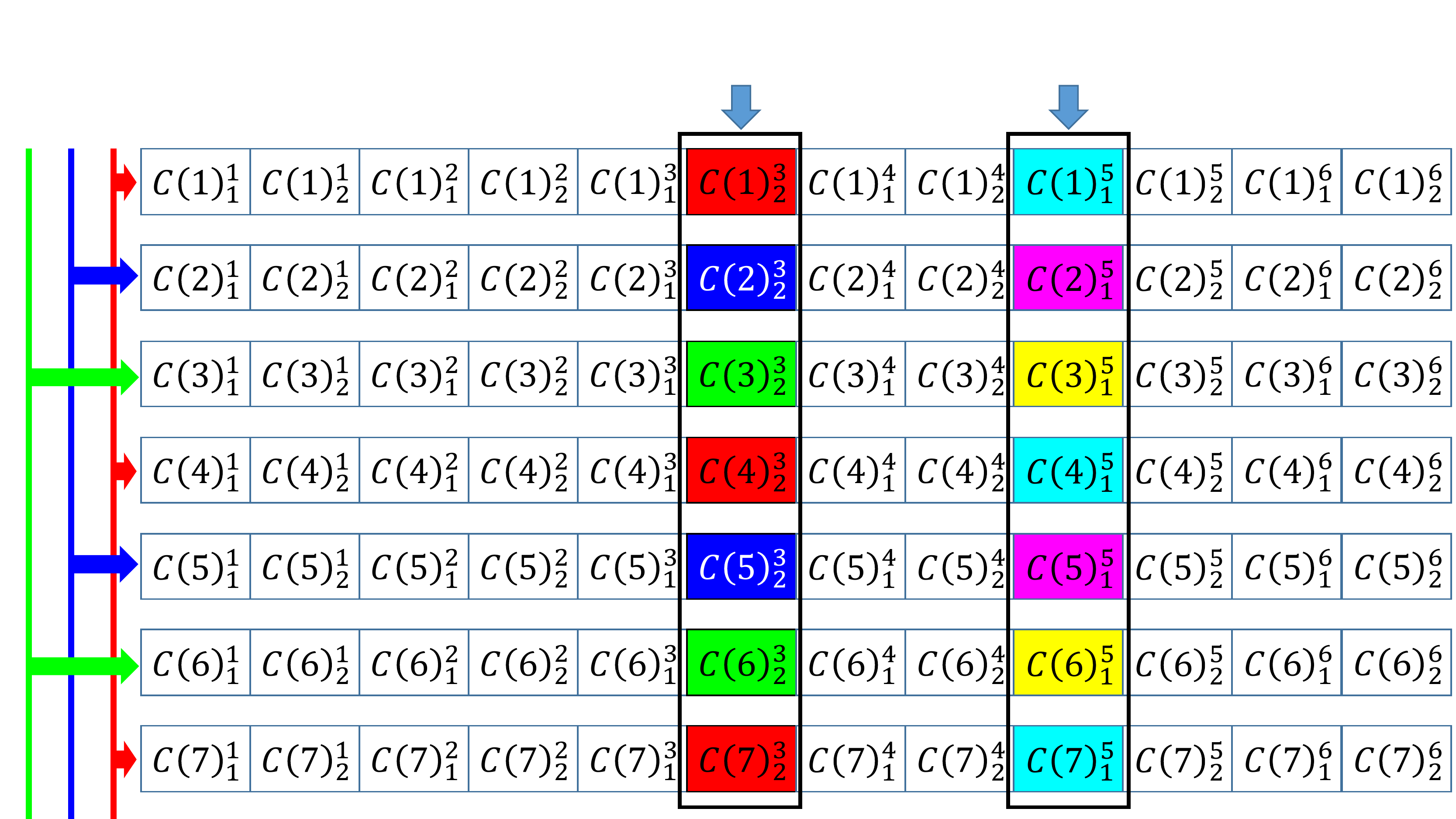}
\caption{ \label{fig:combine-colors}
An example with $7$ classes. Initially, $\mathtt{FlowDec}$ uses 7 colors in each column (one color per class). Next, $\mathtt{ColOptimize}$, called with parameter $\tau=3$, combines the colors resulting in the use of $3$ colors per column.
}
\end{center}
\end{figure}

\begin{figure}[h]
\begin{center}
\includegraphics[scale=0.31]{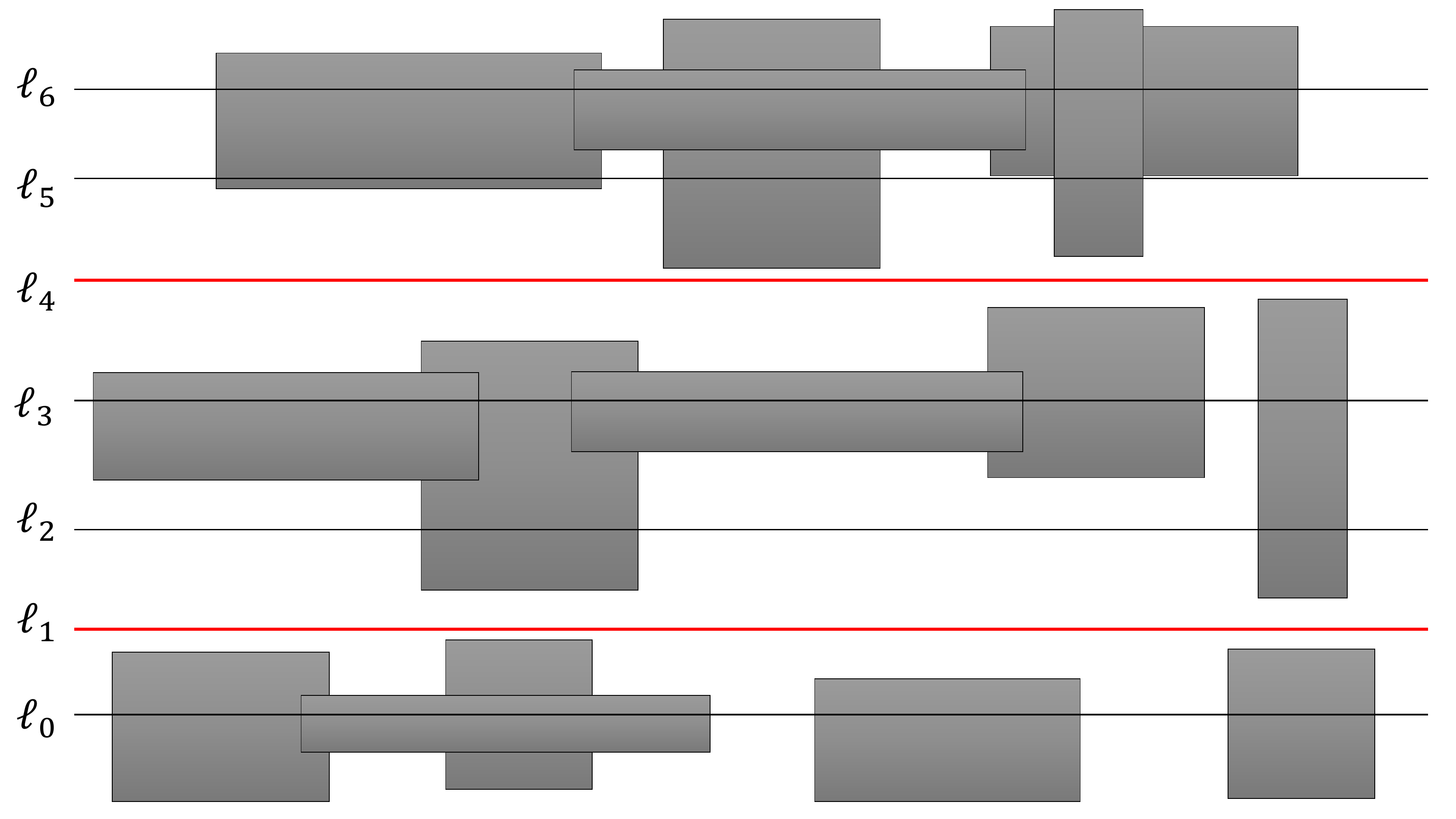}
\caption{An example for rectangle collection $\cR_0$.
The red lines are the ones whose index $i$ satisfy $i\equiv 1\mod 3$. None of the rectangles in $\cR_0$ is intersected by a red line.
\label{fig:cR0}
}
\end{center}
\end{figure}

\clearpage
\section{Missing algorithms pseudocode}\label{Appendix:sec:missing codes}
\IncMargin{1em}
\begin{algorithm}[H]
\SetKwInOut{Input}{input}\SetKwInOut{Output}{output}
\Input{A set of flows $F$}
\Output{A partition of $F$ into feasible subsets}
\BlankLine
$i \leftarrow 1$\;
\While{$F\neq \emptyset$}{
$(C^i_1,C^i_2) \leftarrow \mathtt{rCOVER}(F)$\;
$i\leftarrow i+1$\;
$F\leftarrow F \backslash (C^i_1 \cup C^i_2)  $\;
}
\Return $\{C_i^j, C_2^j: 1 \leq j < i\}$\;
\caption{$\mathtt{FlowDec}$}\label{algo:flowdec}
\end{algorithm}\DecMargin{1em}

\vspace{3mm}

\IncMargin{1em}
\begin{algorithm}[H]
\SetKwInOut{Input}{input}\SetKwInOut{Output}{output}
\Input{A rectangle $R\in\cR$}
\Input{The last state of $\algtwoSp$; a 2-representative-line set $L_{\cR}$ for $\cR$}
\Output{A color for $R$}
\BlankLine
$y\leftarrow T(R) \mod 3$\;
\Return $\mathtt{KT}(\cR_y,R)$\;
\caption{$\algtwoSp$}\label{algo:col2sp}
\end{algorithm}\DecMargin{1em}

\vspace{3mm}

\IncMargin{1em}
\begin{algorithm}[H]
\SetKwInOut{Input}{input}\SetKwInOut{Output}{output}
\Input{A flow $f$}
\Input{The last state of $\mathtt{ProcLarges}$; a capacitated line graph $P=(V,E,c)$}
\Output{A round for $f$}
\BlankLine
$L\leftarrow \emptyset$\;
\For{$i\leftarrow 1$ \KwTo $\lceil \log_{4/3} \cmax \rceil + 1$}{
$y(\ell_i) \leftarrow (3/4)^i\cdot\cmax$\;
$L\leftarrow L \cup \{\ell_i\}$\;
}
Construct a rectangle $\cR(f)$\;
$\mathtt{RectCol(\cR(f)),L}$\;
\Return the color index of $\cR(f)$ \;
\caption{$\mathtt{ProcLarges}$}\label{algo:proclarges}
\end{algorithm}\DecMargin{1em}

\vspace{3mm}

\IncMargin{1em}
\begin{algorithm}[H]
\SetKwInOut{Input}{input}\SetKwInOut{Output}{output}
\Input{A flow $f$}
\Input{The last state of $\mathtt{SolveRUFPP}$; a capacitated line graph $P=(V,E,c)$}
\Output{A round for $f$}
\BlankLine
\lIf{$\sigma_f \geq (1/4) b_f$}{
$\mathtt{ProcLarges(f,P)}$}
\lElse{$\mathtt{ProcSmalls(f,P)}$}
\Return;
\caption{$\mathtt{SolveRUFPP}$}\label{algo:rufp}
\end{algorithm}\DecMargin{1em}
\end{document}